\documentclass[onecolumn]{article}

\pdfoutput=1
\usepackage[utf8]{inputenc}
\usepackage[english]{babel}
\usepackage[T1]{fontenc}
\usepackage{hyperref}
\usepackage{authblk}
\usepackage[a4paper, margin=1in]{geometry}

\usepackage[braket]{qcircuit}

\usepackage{amsmath,amssymb,amsthm,bm,mathtools,amsfonts,mathrsfs,bbm,dsfont, physics}
\usepackage[shortlabels]{enumitem}

\usepackage{mathpazo} % different font

\usepackage{tikz}
\usepackage{float}
\usepackage[center]{caption}
\usepackage{subcaption}
\usetikzlibrary{positioning}
\usetikzlibrary{shapes}
\usepackage{changepage}
\usepackage{graphicx}

\usepackage[framemethod=TikZ]{mdframed} %for boxes
\usepackage{thm-restate} % for restating

\usepackage{empheq}

\usepackage{relsize}

\usepackage[capitalize]{cleveref}

\DeclareRobustCommand{\[}{\begin{equation}}
\DeclareRobustCommand{\]}{\end{equation}}

\newcommand{\id}{\ensuremath{\mathds{1}}}

\declaretheoremstyle[
    headfont=\bfseries, 
    bodyfont=\normalfont,
    headpunct={.},
    spacebelow=\parsep,
    spaceabove=\parsep,
    mdframed={
      roundcorner=10pt,
      linewidth=1pt,
        innertopmargin=6pt,
        innerbottommargin=6pt, 
        skipabove=3ex, 
        skipbelow=3ex,
         nobreak=true
       } 
]{framedstyle}

\declaretheorem[style=framedstyle,name=Theorem]{theorem}

\declaretheorem[style=framedstyle,name=Lemma, numberlike=theorem]{lemma}
\declaretheorem[style=framedstyle,name=Definition, numberlike=theorem]{definition}

   \usepackage{regexpatch}
\makeatletter
\xpatchcmd\thmt@restatable{%
\csname #2\@xa\endcsname\ifx\@nx#1\@nx\else[{#1}]\fi
}{%
\ifthmt@thisistheone
\csname #2\@xa\endcsname\ifx\@nx#1\@nx\else[{#1}]\fi
\else
\csname #2\@xa\endcsname[{restated}]
\fi}{}{}
\makeatother

\begin{document}

\title{\textbf{Testing multipartite productness is easier \\ than testing bipartite productness}}
\author[1,2,3]{Benjamin D.M. Jones}
\author[1,4]{Ashley Montanaro}

\affil[1]{School of Mathematics, University of Bristol, UK.}
\affil[2]{H. H. Wills Physics Laboratory, University of Bristol, UK.}
\affil[3]{Quantum Engineering Centre for Doctoral Training,  University of Bristol, UK.}
 \affil[4]{Phasecraft Ltd.}

\maketitle

 \begin{abstract}
     We prove a lower bound on the number of copies needed to test the property of a multipartite quantum state being product across some bipartition (i.e. not genuinely multipartite entangled), given the promise that the input state either has this property or is $\epsilon$-far in trace distance from any state with this property. We show that $\Omega(n / \log n)$ copies are required (for fixed $\epsilon \leq \frac{1}{2}$), complementing a previous result that $O(n / \epsilon^2)$ copies are sufficient. Our proof technique proceeds by considering uniformly random ensembles over such states, and showing that the trace distance between these ensembles becomes arbitrarily small for sufficiently large $n$ unless the number of copies is at least $\Omega (n / \log n)$. We discuss implications for testing graph states and computing the generalised geometric measure of entanglement.
 \end{abstract}

\tableofcontents

 \section{Introduction}

Quantum entanglement \cite{horodecki2009quantum,guhne2009entanglement, bengtsson2017geometry} is celebrated as a ubiquitous resource across the whole landscape of quantum information and technology. In measurement based approaches to quantum computation \cite{raussendorf2001one,jozsa2006introduction, briegel2009measurement}, one seeks to generate entanglement between multiple sites, for example via the creation of graph states \cite{hein2006entanglement, hein2004multiparty}, and an important practical task is to be able to certify the presence (or lack of) such entanglement \cite{markham2020simple, jungnitsch2011taming}. Multipartite entanglement \cite{walter2016multipartite, bengtsson2016brief} is also a fundamental component of quantum networks \cite{kimble2008quantum, tavakoli2021bell} and plays a significant role in quantum error correction \cite{scott2004multipartite, dur2007entanglement}.

In classical computer science, the domain of property testing \cite{goldreich1998property, fischer2004art} seeks to ascertain if a given object has some property $P$, or is far away from having that property.  An \textit{$\epsilon$-tester} takes as input either $x \in P$ or $x$ $\epsilon$-far from $P$, and in the former case it accepts with probability at least $\frac{2}{3}$, whereas in the latter case it accepts with probability at most $\frac{1}{3}$. A tester is deemed \textit{efficient} if the number of queries made (e.g. number of bits of the object read) is much less than the size $n$ of the object. \textit{Quantum} property testing applies these notions to the quantum world, where one can take either the tester or the object to be tested (or both) to be quantum mechanical in some aspect -- see \cite{montanaro2013survey} for a comprehensive review. When testing properties of quantum states, one typically seeks algorithms that minimise the number of copies required to test the desired property. In particular, it is highly desirable to prove lower bounds on the number of copies required, to understand the optimality of various approaches and the fundamental limits on extracting information from quantum states.

In this paper we will study the property of a multipartite quantum state being product across some (unknown) bipartition, or equivalently the property of not being \textit{genuinely multipartite entangled}, through the lens of property testing. Let us formalise some definitions. Recall that a bipartite pure state $\ket{\psi} \in \mathbbm{C}^{d_1}\otimes \mathbbm{C}^{d_2}$ is \textit{entangled} if it cannot be written as a product state, i.e. as $\ket{\psi} = \ket{\phi}\otimes \ket{\tau}$ for some states $\ket{\phi} \in \mathbbm{C}^{d_1}$ and $\ket{\tau} \in \mathbbm{C}^{d_2}$.

\begin{definition}
    Consider a pure quantum state $\ket{\psi} \in ( \mathbbm{C}^d ) ^{\otimes n}$ consisting of $n$ parties, each of local dimension $d$. We say that it is
    \begin{itemize}
        \item \textit{Genuinely multipartite entangled} (GME) if it is entangled across any bipartition of the $n$ parties.
        \item \textit{Bipartite product} (BP) if it is not GME, that is, there exists some non-trivial partition $S \subset [n]$ such that the state is product across this bipartition.
\item \textit{Multipartite product} (MP) if the state is product across every bipartition, i.e. the state can be written as the tensor product of $n$ local states.
    \end{itemize} 
   
\end{definition}

In \cite{harrow2013testing}, it is shown that given an $n$-partite state $\ket{\psi}$ that is either multipartite product, or is at least $\epsilon$-far from any multipartite product state, there exists a tester using two copies of the input state $\ket{\psi}$, and accepts with certainty if $\ket{\psi}$ is MP and accepts with probability at most $1-\Theta(\epsilon^2)$ otherwise. Repeating this procedure $k$ times (using $2k$ copies) reduces this latter probability to $(1-\Theta(\epsilon^2))^{k} \leq e^{- \Theta (k \epsilon^2)}$, and hence the property of multipartite productness can be tested using $O(1/\epsilon^2)$ copies, for any $n$. The proof strategy uses the product test \cite{mintert2005concurrence}, which in turn consists of applying the swap test \cite{buhrman2001quantum} (a simple test for equality of two states) across corresponding pairs of subsystems of two copies of the input state (also see \cite{montanaro2013survey} for a proof sketch).

Furthermore, a general result is derived in \cite{harrow2017sequential} for testing multiple properties of a quantum state simultaneously. More specifically, given a set of measurement operators $\Lambda_i$ (POVM elements satisfying $0 \leq \Lambda_i \leq \id$ ) and an input state $\rho$ with the promise that either $\text{Tr}(\Lambda_i \rho)$ is small for all $i$, or there exists at least one $i$ with $\text{Tr}(\Lambda_i \rho)$ large, the authors construct a procedure that distinguishes between these cases using one copy of the input state $\rho$. In the same paper this result is applied to testing bipartite productness: building upon the result from \cite{harrow2013testing} they derive a tester for the property of being bipartite product using $O(n/\epsilon^2)$ copies of the state.

In this work, we show that this is close to optimal --  at least $\Omega(n / \log n)$ copies are needed to test bipartite productness, for any fixed constant $0 < \epsilon \leq \frac{1}{2}$. To the best of our knowledge this is the first lower bound constructed for this problem. Our main result can be stated formally as follows.

\begin{restatable}{theorem}{mainthm}
    \label{thm:main}
    An $\epsilon$-tester for testing the property of a multipartite state $\ket{\psi}\in (\mathbbm{C}^{d})^{\otimes{n}}$ being bipartite product requires at least $\Omega \left (n / \log n \right )$ copies of the input state $\ket{\psi}$, for any $0 < \epsilon \leq \frac{1}{2}$.
\end{restatable}

So testing bipartite productness across an unknown bipartition is harder than testing both multipartite productness or productness across a known partition (both can be done with $O(1/\epsilon^2)$ copies), hence it appears that the uncertainty regarding which partition the state is product across is responsible for the increase in hardness.  We now comment on two initial applications of our result. 

Recall that \textit{graph states} \cite{hein2006entanglement, hein2004multiparty} are defined by associating a qubit initialised in the $\ket{+}$ state for every node, and applying a controlled-$Z$ gate for every corresponding edge. Given a graph state, one can consider testing classical properties of the underlying graph \cite{goldreich2010introduction} using few copies of the state \cite{montanaro2022quantum, zhao2016fast}. In particular, our work here relates to the property of the underlying graph being \textit{connected}: if there exists a path from any vertex to any other vertex. The underlying graph is not connected if and only if the associated state is bipartite product. Therefore our results imply that any attempt to test non-connectivity of the underlying graph by testing if the state is bipartite product must use $\Omega(n / \log n)$ copies. However, it is not ruled out that one could test for non-connectivity using fewer copies, taking advantage of the information that the given state is promised to be a graph state.

As a second application, consider the following quantifier of multipartite entanglement 

\begin{align}
E_G(\ket{\psi}) &:= 1- \max_{\ket{\phi} \text{ is BP}} \abs{\braket{\psi}{\phi}}^2 \\
&= \min_{\ket{\phi} \text{ is BP}} ~ D \left ( \ketbra{\psi}, \ketbra{\phi} \right )^2
\end{align}
for $D$ the trace distance. This is known as the \textit{generalised geometric measure of entanglement} -- see \cite{sen2010channel, das2016generalized, ma2023multipartite} and references therein. Thus we can reinterpret our main result as showing that to determine if either $E_G(\ket{\psi})=0$ or $E_G(\ket{\psi}) \geq \epsilon^2$ (given the promise that one of them holds), one requires $\Omega(n / \log n)$ copies of $\ket{\psi}$, for any $0 < \epsilon\leq \frac{1}{2}$. So in general one can expect computing $E_G$ to require at least this many copies.

Our proof of \cref{thm:main} proceeds in several steps:
\begin{enumerate}[(i)]
    \item We first show in \cref{lem:existenceimplies} that if a tester exists, then this places a lower bound on the trace distance of certain quantum states. These quantum states are respectively close to distributions over BP and $\epsilon$-far from BP states.
    \item We then give an upper bound on the trace distance between these states as a function of $n$ (the number of parties), $k$ (the number of copies) and $d$ (the local dimension) -- this is \cref{lem:main_td}.
    \item Finally, we see that unless $k = \Omega (n / \log n)$, then this upper bound goes to zero, which contradicts the existence of a tester.
\end{enumerate}

The bulk of the technical work is in proving point (ii), which requires calculations involving the Haar measure, symmetric subspace, and permutation matrices -- see e.g. \cite{hayden2006aspects,harrow2013church, mele2023introduction} for relevant literature. 

\begin{figure}
    \centering
    \includegraphics[width=\textwidth]{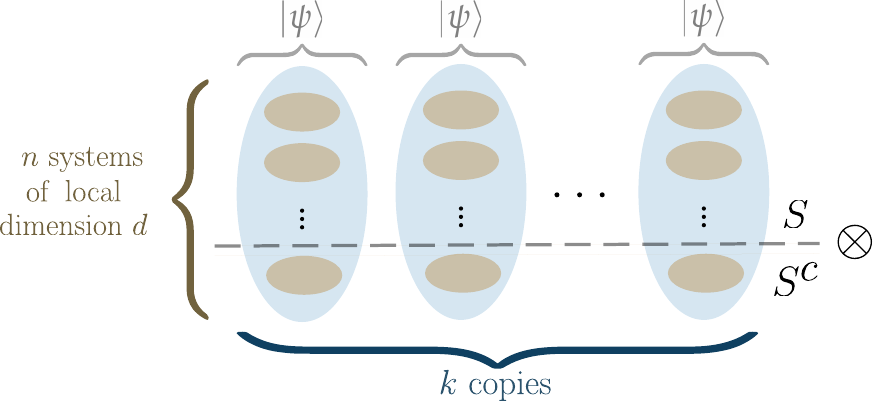}
    \caption{Illustration of the property considered in this paper. The input is given by a quantum state $\ket{\psi} \in (\mathbbm{C}^d)^{\otimes n}$, which is either product across some bipartiton $S:S^c$, or $\epsilon$-far from being product. We are interested in algorithms for distinguishing these cases that use a small number of copies $k$ of the input state $\ket{\psi}$. Some of the technical aspects of our work (e.g. \cref{lem:main_td}) involve unitaries $U_\alpha$ that permute the $k$ systems according to a permutation $\alpha \in \mathcal{S}_k$. Note that in the context of this diagram, these permutations $U_\alpha$ permute the $k$ `columns', and not the $n$ `rows', and hence given some bipartion $S \subset [n]$ we can write $U'_\alpha = U_\alpha \otimes_S U_\alpha$  for $U'_\alpha$ acting on the whole space --  see also \cref{eq:tensor_product_bipartition}.}
    \label{fig:nkd}
\end{figure}

\subsection{Mathematical Preliminaries}

We use $D(\rho, \sigma) := \tfrac{1}{2} \norm{\rho - \sigma}_1$ to denote trace distance, $\binom{n}{k} := \tfrac{n!}{(n-k)!k!}$ the binomial coefficients, and $[n]:=\{0, \dots , n-1\}$ the set of integers from $0$ to $n-1$ inclusive. We write `$\ln$' for the natural logarithm and `$\log$' for the logarithm to base $2$.

Consider $k$ quantum systems of local dimension $d$, i.e. some state $\ket{\psi} \in (\mathbbm{C}^d)^{\otimes k}$. Define unitaries $U_\alpha$ that permute the $k$ systems for some permutation $\alpha$ in the symmetric group $\mathcal{S}_k$: 
\[
U_\alpha \ket{x_1, \dots , x_k} = \ket{x_{\alpha^{-1}(1)}, \dots , x_{\alpha^{-1}(k)}}. \label{eq:U_alpha}
\]
Note that $U_\alpha U_\beta = U_{\alpha \beta}$. One can then define the \textit{symmetric subspace} \cite{harrow2013church}  as follows:
\[
\text{Sym}^k_{d}:=\left \{\ket{\psi} \in (\mathbbm{C}^d)^{\otimes k} \quad : \quad U_\alpha \ket{\psi} = \ket{\psi} \quad  \forall ~ \alpha \in \mathcal{S}_k \right \},
\]
which can equivalently be defined as the span of states of the form $\ket{\psi} = \ket{\phi}^{\otimes k}$ for some $\ket{\phi}\in \mathbbm{C}^d$. We can write the projector $\Pi^k_d$ onto the symmetric subspace as
\[
\Pi^k_d := \mathop{\mathbbm{E}}_{\alpha \in \mathcal{S}_k} \left [ U_\alpha \right ] =  \frac{1}{k!} \sum_{\alpha \in \mathcal{S}_k} U_\alpha. \label{eq:symsub}
\]
Now let $d\psi$ denote the Haar measure on quantum states. Then a well-known fact \cite{harrow2013church, mele2023introduction} is that integration over $k$ copies of a state $\ket{\psi} \in \mathbbm{C}^d$ is proportional to the projector onto the symmetric subspace, specifically we have
\[
{k+d-1 \choose k } \int d\psi \ketbra{\psi}^{\otimes k} = \Pi^k_d . \label{eq:haar_symm_subspace}
\]

We use $S^c$ to denote the complement of a subset $S \subseteq [n]$, and $\abs{S}$ to denote the size of the set $S$, so in particular $\abs{S} + \abs{S^c} = n$. We denote the empty set by $\emptyset$.

Recall that any permutation $\alpha \in \mathcal{S}_k$ can be written as a product of disjoint cycles, which is unique up to reordering. We refer to the number of cycles in this cycle decomposition of a permutation $\alpha$ as the \textit{cycle number}, denoted $c(\alpha)$. For example, the (cyclic) permutation $(1 2 3) \in \mathcal{S}_3$ has cycle number $1$, and the identity permutation $e = (1) \dots (k) \in \mathcal{S}_k$ has cycle number equal to $k$.

For a multipartite quantum state $\ket{\psi} \in (\mathbbm{C}^d)^{\otimes n}$ that is product across some bipartition $S \subset [n]$, we may use labels on the states and tensor product symbol for clarity. For example, if the state $\ket{\psi} \in (\mathbbm{C}^d )^{\otimes n}$ is product across the bipartition $S:S^c$ with respective states $\ket{\phi}$ and $\ket{\tau}$, for $k$ copies we may write (see also \cref{fig:nkd})
\[
\ket{\psi}^{\otimes{k}} = \ket{\phi^S}^{\otimes k} \otimes_S \ket{\tau^{S^c}}^{\otimes k}, \label{eq:tensor_product_bipartition}
\]
and similarly for operators.

Finally, recall that \textit{Schmidt decomposition} allows us to write any bipartite state $\ket{\psi} \in \mathbbm{C}^{d_1} \otimes \mathbbm{C}^{d_2}$ as
\[
\ket{\psi} = \sum_i \gamma_i \ket{v_i} \ket{w_i},
\]
where the \textit{Schmidt coefficients} $\gamma_i$ are non-negative, the sets $\ket{v_i}$ and $\ket{w_i}$ are respectively orthonormal, and the number of terms in the expansion is minimal and referred to as the \textit{Schmidt rank}.

\section{Results}

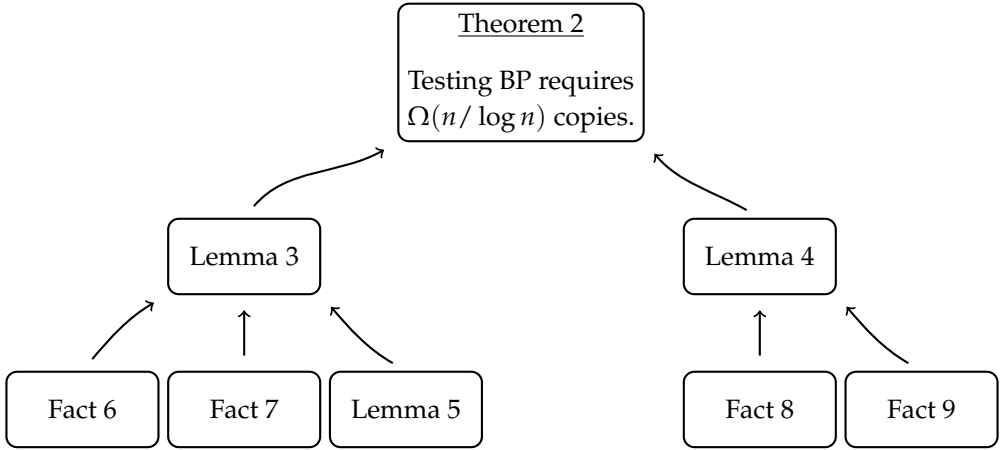
\begin{figure}[H]
    \centering

\begin{tikzpicture}[block/.style={rounded corners, minimum width=2cm, minimum height=1cm, draw}, line/.style={->,shorten >=0.4cm,shorten <=0.4cm},thick]

\node[block] (1) {\shortstack{\underline{\cref{thm:main}} \\[2ex] Testing BP requires \\ 
$\Omega(n /\log n)$ copies.}};

\node[block, below left=1cm and 1cm of 1] (2) {\cref{lem:existenceimplies}};

\node[block, below left=1cm and 0.1cm of 2] (4) {\cref{fact:schmidt}};
\node[block, below =1cm of 2] (5) {\cref{fact:haar}};
\node[block, below right=1cm and 0.1cm of 2] (6) {\cref{lem:schmidtprob}};

\node[block, below right=1cm and 0.5cm of 1] (3) {\cref{lem:main_td}};

\node[block, below =1cm of 3] (7) {\cref{fact:binbounds}};
\node[block, below right=1cm and 0.1cm of 3] (8) {\cref{fact:cycles}};

\draw [->, shorten <= 0.2cm,  shorten >= 0.2cm] (2.north) to[out=50,in=-150] (1.south west);
\draw [->, shorten <= 0.2cm,  shorten >= 0.2cm] (3.north) to[out=150,in=-50] (1.south east);

\draw [->, shorten <= 0.2cm,  shorten >= 0.2cm] (4.north) to[out=50,in=-150] (2.south west);
\draw [->, shorten <= 0.2cm,  shorten >= 0.2cm] (5.north) to (2.south);
\draw [->, shorten <= 0.2cm,  shorten >= 0.2cm] (6.north) to[out=150,in=-50] (2.south east);

\draw [->, shorten <= 0.2cm,  shorten >= 0.2cm] (7.north) to (3.south);
\draw [->, shorten <= 0.2cm,  shorten >= 0.2cm] (8.north) to[out=150,in=-50] (3.south east);

\end{tikzpicture}

\caption{Proof structure and supporting results.}
\label{fig:results}
\end{figure}

We first state the two main ingredients used in the proof of our main result.

\begin{restatable}{lemma}{existenceimplies}
    
\label{lem:existenceimplies}
For $0 < \epsilon\leq \frac{1}{2}$, the existence of an $\epsilon$-tester for the property of a multipartite state being bipartite product using $k$ copies implies that

\begin{align}
     D(\rho , \sigma) ~\geq~  \frac{1}{3}  - O(2^{-n}),
\end{align}
for $D$ the trace distance, and where
\[
\rho = \frac{ \Pi_{d^n}^k}{ {d^n + k -1 \choose k}}, \qquad \qquad
\sigma = 
\mathop{\mathbbm{E}}_{S \subseteq [n]} \bigg [ \frac{ \Pi_{d^{|S|}}^k \otimes_S \Pi_{d^{n-|S|}}^k }{{d^{\abs{S}} + k -1 \choose k} {d^{n - \abs{S}} + k -1 \choose k}} \bigg ],
\]
for $ \Pi_{d}^k$ the projector onto the symmetric subspace of $k$ systems of local dimension $d$.
\end{restatable} 

\begin{restatable}{lemma}{rhosigdist}\label{lem:main_td}
Consider the following states 
\[
\rho = \frac{ \Pi_{d^n}^k}{ {d^n + k -1 \choose k}}, \qquad \qquad
\sigma = 
\mathop{\mathbbm{E}}_{S \subseteq [n]} \bigg [ \frac{ \Pi_{d^{|S|}}^k \otimes_S \Pi_{d^{n-|S|}}^k }{{d^{\abs{S}} + k -1 \choose k} {d^{n - \abs{S}} + k -1 \choose k}} \bigg ],
\]
where $ \Pi_{d}^k$ denotes the projector onto the symmetric subspace of $k$ systems of local dimension $d$. Then their squared trace-distance is upper bounded by the following expression:
    \[
     D(\rho, \sigma)^2 \leq \frac{k!}{4} \bigg (1 + (k!)^3\left (\frac{1+d}{2d} \right )^n -  e^{-k^2 / d^{n}}   \bigg  ).
    \]
\end{restatable}

Using these two ingredients we can prove our main theorem.

\subsection{Proof of Main Result}

\mainthm*
\begin{proof}[Proof of \cref{thm:main}]

If a tester exists, by \cref{lem:existenceimplies} we have that
\begin{align}
    \frac{1}{3} &\leq  O(2^{-n}) + D(\rho , \sigma)
\end{align}

for the states $\rho$, $\sigma$ as stated. \cref{lem:main_td} then gives us

\begin{align}
D(\rho, \sigma)^2 & \leq  \frac{k!}{4} \bigg (1 + (k!)^3\left (\frac{1+d}{2d} \right )^n -  e^{-k^2 / d^{n}}   \bigg  ) \\
& \leq  \frac{k!}{4} \bigg ( (k!)^3\left (\frac{1+d}{2d} \right )^n +  O(k^2 / d^{n})  \bigg  ) \\
&\leq \frac{k^{4k}}{4} 2^{-n} \bigg  (1+\frac{1}{d} \bigg )^n + O(k^{k+2}d^{-n}),
\end{align}
where we used $k! \leq k^k$ and $1-e^{-k^2 / d^n} = O(k^2 / d^{n})$ (assuming $k^2 / d^n < 1 ~ \forall n$)\footnote{The latter can be seen from $1-e^{-f(n)} \leq f(n)  (1 - \frac{f(n)}{2!} + \frac{f(n)^2}{3!} - \dots  )   \leq f(n)  (1  + \frac{1}{3!} + \frac{1}{5!} + \dots  ) =  \sinh(1) f(n) = O(f(n))$ for all functions satisfying $0 \leq f(n) \leq 1$ $\forall n$.}. 

As the local dimension $d \geq 2$, we have $1+\frac{1}{d} \leq \frac{3}{2}$, and so

\begin{align}
D^2 &\leq \frac{k^{4k}}{4} 2^{-n} \left  (\tfrac{3}{2} \right )^n + O( k^{k+2}d^{-n}) \\
 &\leq O\left (k^{4k} \left (\tfrac{3}{4} \right )^n\right ) + O( k^{k+2}2^{-n}) \\
 &\leq O \bigg ( k^{4k} \left (\tfrac{3}{4} \right )^n \bigg ) \\
  &= O \bigg ( 2^{4 k \log k - n \log 4/3} \bigg ).
\end{align}

Thus after taking the square root we have
\begin{align}
    D \leq O \bigg ( 2^{ 2 k \log k -  a n } \bigg ),
\end{align}
where $a=\tfrac{1}{2}\log 4/3 \approx 0.208$. Now take $k < \frac{cn}{\log n}$, with $0 < c < \tfrac{a}{2}$. Then observe that
\begin{align}
    2 k \log k - a n  &< \frac{2cn}{\log n} \bigg ( \log c  + \log n - \log \log n \bigg ) - an \\
    &= (2c-a) n + (c \log c) \left (\frac{n}{\log n} \right) - c  \left (\frac{n \log \log n}{\log n} \right ) \\
    &< 0  \quad \text{for sufficiently large $n$, as $2c < a$.}
\end{align}
This means that $O  ( 2^{2k \log k - a n } )$ would tend to zero as $n$ goes to infinity.
This contradicts the assertion that a tester must satisfy
\begin{align}
    \frac{1}{3} &\leq  O(2^{-n}) + D(\rho , \sigma),
\end{align}
as the right hand side of the inequality would tend to zero as $n$ tends to infinity. Hence no tester exists unless $k\in \Omega(\tfrac{n}{\log n}).$

\end{proof}

We also state and prove the following result on the distribution of the maximum Schmidt coefficient under the Haar measure. Aside from potentially being of independent interest, it primarily serves as a crucial ingredient in \cref{lem:existenceimplies} where it is needed to show that the Haar distribution is close to the same distribution conditioned on states with bounded maximum Schmidt coefficient.

\begin{lemma} \label{lem:schmidtprob}
Let $\ket{\psi} \in (\mathbbm{C}^d)^{\otimes n}$ be drawn uniformly at random from the  Haar measure, and let $\Gamma_{\max}$ denote the maximum Schmidt coefficient over all non-trivial bipartitions. 

Let $\frac{\sqrt{3}}{2} \leq \gamma < 1$ be a constant. Then there exist positive constants $c_1$, $c_2$ and $N$ (in terms of $\gamma$ and $d$, expressions given below) such that for all $n > N$
\[
\mathbbm{P} \left (\Gamma_{\max} > \gamma \right ) \leq c_1 2^{n} e^{- c_2 d^n}.
\]
Here $c_1, c_2, N$ are  given by
\begin{align}
    c_1 &=  \tfrac{1}{2} \left ( \frac{30}{\gamma^2} \right )^{2d}, \\
    c_2 &= \frac{d \gamma^4}{126 \ln 2}, \\
     N &=  \frac{1}{\ln d} \ln ( \frac{252 \ln 2 \ln (\frac{30}{\gamma^2})}{\gamma^4} ). 
\end{align}

\end{lemma}
\begin{proof}
    From \cite{hayden2006aspects, harrow2004superdense} we have for $\lambda_{\max}$ the maximum eigenvalue of either reduced density matrix of a Haar random state:
    \[
    \mathbbm{P} \left (\lambda_{\text{max}} > \frac{1+\delta}{d_A} \right ) \leq \left ( \frac{10 d_A}{\delta} \right )^{2d_A} \exp( -d_B \frac{\delta^2}{14\ln 2}),
    \]
where $d_A$ and $d_B$ are the local dimensions across a fixed bipartition. Recall that the Schmidt coefficients of a bipartite pure state are equal to the square roots of the eigenvalues of either reduced density matrix. We then use the relabelling $\frac{1 + \delta}{d_A} = \gamma^2 \iff \delta = d_A \gamma^2 -1$ to write the above as

\begin{align}
\mathbbm{P} \left (\lambda_{\max} > \gamma^2 \right ) &= \mathbbm{P} \left (\gamma_{\max} > \gamma \right ) \\ 
&\leq \left ( \frac{10 d_A}{d_A \gamma^2-1} \right )^{2d_A} \exp( -d_B \frac{(d_A \gamma^2-1)^2}{14\ln 2}),
\end{align}
for $\gamma_{\max}$ the maximum Schmidt coefficient across this bipartition.

As we are considering $n$ parties each with local dimension $d$, set $d_A=d^x$, where $x \leq \tfrac{n}{2}$, and $d_B = d^{n-x}$. Note also that for $\gamma \geq \frac{\sqrt{3}}{2}$ and $d_A \geq 2$, we have

\begin{alignat}{3}
d_A \gamma^2  \geq ~ \frac{3}{2} \qquad \iff \qquad   d_A\gamma^2 -1 \geq \frac{d_A\gamma^2}{3}.
\end{alignat}

Hence we can write

\begin{align}
   \mathbbm{P} \left (\gamma_{\max} > \gamma\right ) &\leq \left ( \frac{30 }{\gamma^2} \right )^{2\cdot d^x} \exp( - \frac{d^{n+x}\gamma^4}{126 \ln 2})  \\
        & = \exp(d^x (a - b d^{n})),
\end{align}
for positive constants 
\begin{align}
    a = 2 \ln ( \frac{30}{\gamma^2}), \qquad\qquad
    b = \frac{\gamma^4}{126 \ln 2}.
\end{align}

For $a - b d^n \leq 0$, the worst case is for $x=1$ (when $d^x$ is smallest, and $d^x (a - b d^{n})$ is largest). So for $n$ sufficiently large we have that
\[
\exp(d^x (a - b d^{n})) \leq \exp(d (a - b d^{n})) .
\]

Taking the union bound over all $2^{n-1}-1$ nontrivial bipartitions then gives

\begin{align}
\mathbbm{P} \left (\Gamma_{\max} > \gamma \right ) &\leq 2^{n-1} \exp(d (a - b d^{n}))  \\
&\equiv c_1 2^n e^{-c_2 d^n},
\end{align}
for $\Gamma_{\max}$ the maximum Schmidt coefficient over all bipartitions, and where
\begin{align}
    c_1 = \tfrac{1}{2} e^{d a} =  \tfrac{1}{2} \left ( \frac{30}{\gamma^2} \right )^{2d}, \qquad\qquad
    c_2 =  d b = \frac{d \gamma^4}{126 \ln 2}.
\end{align}

Finally, note that we can rewrite the condition
$a - b d^n \leq 0$ as

\begin{align}
     n \geq \frac{\ln (\frac{a}{b})}{\ln(d)} &= \frac{\ln (\frac{\ln 2c_1 }{c_2})} {\ln(d)} = \frac{1}{\ln d} \ln ( \frac{252 \ln 2 \ln (\frac{30}{\gamma^2})}{\gamma^4} ). 
\end{align}

\end{proof}

\subsection{Proof of \cref{lem:existenceimplies}}

\existenceimplies*
    
\begin{proof}
Suppose there is a tester for the property of being bipartite product (BP) using $k$ copies of the input state. This means that there exists an operator (a POVM element) $M : (\mathbbm{C}^d) ^{\otimes kn} \rightarrow (\mathbbm{C}^d )^{\otimes kn}$ with $0 \leq M \leq \id$, such that for all inputs $\ket{\psi} \in (\mathbbm{C}^d )^{\otimes n}$ we have the following.
\begin{itemize}
    \item If $\ket{\psi}$ is BP then $\text{Tr}(M ~ \ketbra{\psi}^{\otimes k})\geq \frac{2}{3}$.
    \item If $\ket{\psi}$ is $\epsilon$-far from being BP then $\text{Tr}(M ~ \ketbra{\psi}^{\otimes k})\leq \frac{1}{3}$.
\end{itemize}

This implies that for any $\ket{\psi}$ that is BP, and any $\ket{\phi}$ that is $\epsilon$ far from being BP we have
\[
\text{Tr} \left (M \left ( \ketbra{\psi}^{\otimes k} - \ketbra{\phi}^{\otimes k}\right ) \right) \geq \frac{1}{3}.
\]
By linearity, this must also hold if we replace the states with averages, respectively according to any distribution $\mathcal{D}_{BP}$ on BP states, and any distribution $\mathcal{D}_F$ on states $\epsilon$-far from being BP. The variational characterisation of the trace distance then also allows us to write
\begin{align}
    \frac{1}{3} &\leq  \text{Tr} \left (M \left ( \mathbbm{E}_{\psi \sim \mathcal{D}_{BP}}(\ketbra{\psi}^{\otimes k}) - \mathbbm{E}_{\phi \sim \mathcal{D}_{F}}(\ketbra{\phi}^{\otimes k})\right ) \right), \\
    &\leq D \bigg (\mathbbm{E}_{\psi \sim \mathcal{D}_{BP}}(\ketbra{\psi}^{\otimes k}) ~ , ~ \mathbbm{E}_{\phi \sim \mathcal{D}_{F}}(\ketbra{\phi}^{\otimes k}) \bigg ).
\end{align}

We will take $\mathcal{D}_{BP}$ as the distribution defined by taking a random non-trivial bipartition of the $n$ parties, and then randomising over pure states on each subsystem using the Haar measure. More concretely, for some subset $S \subseteq [n]$ denote the normalised states
\begin{align}
\tau_S &= \bigg ( \int d \theta \ketbra{\theta}^{\otimes k} \bigg ) \otimes_S \bigg ( \int d \omega \ketbra{\omega}^{\otimes k} \bigg ) \\
&= \frac{ \Pi_{d^{|S|}}^k \otimes_S \Pi_{d^{n-|S|}}^k }{{d^{\abs{S}} + k -1 \choose k} {d^{n - \abs{S}} + k -1 \choose k}},
\end{align}

where $\Pi^k_d$ is the projector onto the symmetric subspace -- see \cref{eq:symsub} and \cref{eq:haar_symm_subspace}. Hence as a distribution over BP states we take 
\[
\sigma' := \mathbbm{E}_{\psi \sim \mathcal{D}_{BP}}(\ketbra{\psi}^{\otimes k}) = \mathop{\mathbbm{E}}_{\substack{S \subseteq [n]\\ S \neq \emptyset, [n]}} \bigg [\tau_S \bigg ].
\]

This state is $O(2^{-n})$ close in trace distance to the state $\sigma$ that includes the trivial bipartions in the average, as seen by the following calculation (using the triangle inequality).

\begin{align}
\norm{\sigma - \sigma'}_1 &=
    \norm{\mathop{\mathbbm{E}}_{S \subseteq [n]} \bigg [ \tau_S \bigg ] - \mathop{\mathbbm{E}}_{\substack{S \subseteq [n]\\ S \neq \emptyset, [n]}} \bigg [\tau_S \bigg ]}_1 \\
    &= \norm{\frac{1}{2^n} \left (\tau_\emptyset + \tau_{[n]} \right ) +  \left (\frac{1}{2^n} - \frac{1}{2^{n}-2} \right )\sum_{{\substack{S \subseteq [n]\\ S \neq \emptyset, [n]}}} \tau_S }_1 \\
    &\leq \frac{1}{2^n}  \left ( \norm{\tau_\emptyset}_1 + \norm{\tau_{[n]}}_1 \right )  +  \abs{\frac{1}{2^n} - \frac{1}{2^{n}-2} }\sum_{{\substack{S \subseteq [n]\\ S \neq \emptyset, [n]}}} \norm{\tau_S }_1 \\
    &= \frac{1}{2^{n-1}} + (2^{n}-2)\abs{ \frac{1}{2^n} - \frac{1}{2^{n}-2}} \\
    &= \frac{1}{2^{n-2}}.
\end{align}

Now define $\mathcal{D}_F$ to be the Haar measure conditioned on the maximum Schmidt coefficient over all bipartitions being at most $\gamma = \sqrt{1-\epsilon^2}$. This guarantees that the output is at least $\epsilon$-far in trace distance from being BP by the following facts, with proof in \cref{app:schmidt}.

\begin{restatable}{fact}{schmidt}\label{fact:schmidt} ~

\begin{enumerate}[(i)]
    \item The maximum Schmidt coefficient $\gamma_{\max}$ of a bipartite state $\ket{\psi} \in \mathbbm{C}^{d_1} \otimes \mathbbm{C}^{d_2}$ is equal to 
    \[
    \max_{\substack{\ket{\alpha} \in \mathbbm{C}^{d_1} \\ \ket{\beta} \in \mathbbm{C}^{d_2}}}\abs{\bra{\psi} ~ \ket{\alpha} \ket{\beta}}.
    \]
    \item If a multipartite state $\ket{\psi} \in (\mathbbm{C}^d)^{\otimes n}$ has maximum Schmidt coefficient at most $\gamma$ across any nontrivial bipartition, then it must be at least $\epsilon$-far in trace distance from any bipartite product state, for $\epsilon = \sqrt{1 - \gamma^2}$.
    \end{enumerate}
\end{restatable}

We also require the following fact, which intuitively states that if two distributions only disagree on a subset that occurs with small probability, then the distributions themselves will be close. We give proof in \cref{app:haar}.

\begin{restatable}{fact}{haarcloseness}\label{fact:haar}
    Let $H$ denote the Haar distribution, and $H_S$ be the Haar distribution conditioned on states belonging to some measurable set $S$. Let $p$ be the probability that a Haar random state does not belong to $S$, i.e. $p = 1-\int_{S} d\psi$.  Define the states

\begin{align}
    \rho & = \mathbbm{E}_{\psi \sim H}(\ketbra{\psi}^{\otimes k}), \\
    \rho' & = \mathbbm{E}_{\phi \sim H_S}(\ketbra{\phi}^{\otimes k}).
\end{align}

Then the trace distance between these states is at most $p$:
\[
D(\rho, \rho') \leq p.
\]

\end{restatable}

Now take $S$ as the set of states with maximum Schmidt coefficient at most $\gamma$. By \cref{lem:schmidtprob}, for $\frac{\sqrt{3}}{2} \leq \gamma \leq 1$ the probability that a Haar random state has maximum Schmidt coefficient greater than $\gamma$ is  at most $c_1 2^{n} e^{- c_2 d^n}$, where $c_1$ and $c_2$ are given in \cref{lem:schmidtprob}. Hence by \cref{fact:haar} the trace distance between the following states
\begin{align}
    \rho & = \mathbbm{E}_{\psi \sim H}(\ketbra{\psi}^{\otimes k}) \\
    \rho' & = \mathbbm{E}_{\phi \sim H_S}(\ketbra{\phi}^{\otimes k})
\end{align}
is at most $c_1 2^{n} e^{- c_2 2^n}$. Finally, \cref{fact:schmidt} tells us that if $\ket{\psi}$ has maximum Schmidt coefficient at most $\gamma$, then it is $\epsilon = \sqrt{1-\gamma^2}$ far in trace distance from any BP state. The condition $\frac{\sqrt{3}}{2} \leq \gamma \leq 1$ is equivalent to $0 \leq \epsilon \leq \frac{1}{2}$.

Thus in summary, by two applications of the triangle inequality the existence of an $\epsilon$-tester for bipartite productness, for $0 < \epsilon \leq \frac{1}{2}$, implies that

\begin{align}
    \frac{1}{3} &\leq D(\rho', \sigma') \\
    &\leq D(\sigma, \sigma') + D(\rho, \rho') + D(\rho, \sigma) \\
    &\leq  O(2^{-n}) + O(2^n e^{-c d^n}) + D(\rho , \sigma) \\
    &\leq  O(2^{-n}) + D(\rho , \sigma),
\end{align}
for $c>0$ a constant and where
\[
\rho = \frac{ \Pi_{d^n}^k}{ {d^n + k -1 \choose k}}, \qquad \qquad
\sigma = 
\mathop{\mathbbm{E}}_{S \subseteq [n]} \bigg [ \frac{ \Pi_{d^{|S|}}^k \otimes_S \Pi_{d^{n-|S|}}^k }{{d^{\abs{S}} + k -1 \choose k} {d^{n - \abs{S}} + k -1 \choose k}} \bigg ].
\]

\end{proof}

\subsection{Proof of \cref{lem:main_td}}

 \rhosigdist*

\begin{proof}

First, we use the following standard inequality to replace the 1-norm with the 2-norm, for any matrix $A\in\mathbbm{C}^{d \times d}$ 
\begin{align}
        \norm{A}_1 &\leq \sqrt{d} \norm{A}_2,
    \end{align}
    where $\norm{A}_p := \text{Tr} \left ( \abs{A}^p \right )^\frac{1}{p}$. So we have
\begin{align}
D(\rho, \sigma)^2 &= \frac{1}{4} \norm{\rho - \sigma}_1^2 \\
&\leq \frac{d^{nk}}{4} \norm{\rho - \sigma}_2^2 
= \frac{d^{nk}}{4} \text{Tr}\left ( (\rho - \sigma)^2\right ) \\
&= \frac{d^{nk}}{4} \bigg (\text{Tr} (\rho^2) + \text{Tr}(\sigma^2) - 2 \text{Tr}(\rho\sigma) \bigg ).
\end{align}
To calculate $\text{Tr}(\rho\sigma)$, we will now use the fact that
    \[
    \text{Sym}^k_{d_1} \otimes \text{Sym}^k_{d_2} \subseteq \text{Sym}^k_{d_1 d_2}.
    \]

To see this, take a state $\ket{\psi} \in \text{Sym}^k_{d_1} \otimes \text{Sym}^k_{d_2}$. Then by definition it is preserved under $U_\alpha \otimes U_\beta$ $\forall \alpha, \beta \in \mathcal{S}_k$, for $U_\alpha$ as defined in \cref{eq:U_alpha}. In particular, it is preserved when $\alpha=\beta$, and we have $U_\alpha \otimes U_\alpha = U'_\alpha$, where $U'_\alpha$ acts on the whole space. So $\ket{\psi}$ is in $\text{Sym}^k_{d_1 d_2}$ -- see also \cref{fig:nkd} for a visual aid. This implies the following relationship between the projectors onto these spaces
  \[
    \Pi^k_{d_1 d_2} \cdot \left ( \Pi^k_{d_1} \otimes \Pi^k_{d_2} \right ) =  \Pi^k_{d_1} \otimes \Pi^k_{d_2}.
    \]

Hence we have

\begin{align}
\text{Tr}(\rho \sigma) &= \mathop{\mathbbm{E}}_{S \subseteq [n]} \bigg [ \frac{\text{Tr} \left ( \Pi_{d^n}^k \cdot \Pi_{d^{|S|}}^k \otimes_S \Pi_{d^{n-|S|}}^k  \right ) }{{d^n + k -1 \choose k} {d^{\abs{S}} + k -1 \choose k} {d^{n - \abs{S}} + k -1 \choose k}} \bigg ] \\
&= \frac{1}{{d^n + k -1 \choose k} }  \mathop{\mathbbm{E}}_{S \subseteq [n]} \bigg [ \frac{\text{Tr} \left ( \Pi_{d^{|S|}}^k \otimes_S \Pi_{d^{n-|S|}}^k  \right )}{ {d^{\abs{S}} + k -1 \choose k} {d^{n - \abs{S}} + k -1 \choose k}} \bigg ]  \\
&= \frac{1}{{d^n + k -1 \choose k} } = \text{Tr}(\rho^2).
\end{align}

So altogether at this stage we have
\begin{align}
 D^2 &\leq \frac{d^{nk}}{4} \bigg (\text{Tr} (\rho^2) + \text{Tr}(\sigma^2) - 2 \text{Tr}(\rho\sigma) \bigg ) \\
&= \frac{d^{nk}}{4} \bigg (\text{Tr} (\sigma^2)  -  \text{Tr}(\rho^2) \bigg ).
\end{align}

We will now use the following fact to help bound $\text{Tr}(\rho^2)$ and $\text{Tr}(\sigma^2)$, deferring the proof to \cref{app:binbounds}.

\begin{restatable}{fact}{binbounds}\label{fact:binbounds} For all $a,b \in \mathbbm{N}$ with $b\geq 1$, we have 
    \[
   \frac{a^b}{b!}\leq  {a + b -1 \choose b}  \leq \frac{a^b }{b!} ~ e^{b^2/a}.
    \]
\end{restatable}

We can use this to bound $\text{Tr}(\rho^2)$ as follows.

\[
\text{Tr}(\rho^2) = \frac{1}{{d^n + k -1 \choose k}} \geq \frac{ k!}{e^{k^2 / d^n} d^{nk}}.
\]

To bound $\Tr(\sigma^2)$, we can again employ \cref{fact:binbounds} to obtain
\begin{align}
\frac{1}{{d^{\abs{S}} + k -1 \choose k}} \leq  \frac{k!}{d^{\abs{S}k}},
\end{align}

so that
\begin{align}
    \text{Tr}(\sigma^2) &= \text{Tr} \left ( \mathop{\mathbbm{E}}_{S, T \subseteq [n]} \left [ \frac{ \left (\Pi_{d^{|S|}}^k \otimes_S \Pi_{d^{|S^c|}}^k \right ) \left ( ~  \Pi_{d^{|T|}}^k \otimes_T \Pi_{d^{|T^c|}}^k \right )}{{d^{\abs{S}} + k -1 \choose k} {d^{\abs{S^c}} + k -1 \choose k} {d^{\abs{T}} + k -1 \choose k} {d^{\abs{T^c}} + k -1 \choose k}} \right ] \right ) \\
    &\leq \frac{(k!)^4}{d^{2nk}} ~  F(k,n, d),
\end{align}
using $\abs{S} + \abs{S^c} + \abs{T} + \abs{T^c} = 2n$ and where we define
\[
F \equiv F(k,n,d) = \text{Tr} \left ( \mathop{\mathbbm{E}}_{S, T \subseteq [n]} \left [ \left ( \Pi_{d^{|S|}}^k \otimes_S \Pi_{d^{|S^c|}}^k ~ \right ) \left (  \Pi_{d^{|T|}}^k \otimes_T \Pi_{d^{|T^c|}}^k \right ) \right ]  \right ).
\]

Hence at this stage we have
\begin{align}
D^2 &\leq  \frac{d^{nk}}{4} \bigg (\text{Tr} (\sigma^2)  -  \text{Tr}(\rho^2) \bigg ) \\ 
&\leq  \frac{d^{nk}}{4} \bigg ( F \cdot \frac{(k!)^4}{d^{2nk}} -  \frac{ k!}{e^{k^2 / d^n} d^{nk}} \bigg ) \\ 
&= \frac{k!}{4} \bigg ( F \cdot \frac{(k!)^3}{d^{nk}} -  e^{-k^2 / d^{n}}   \bigg  ). \label{eq: F expression}
\end{align}

We now seek an upper bound on $F$. Recall that
\[
\Pi^k_d = \mathop{\mathbbm{E}}_{\alpha \in \mathcal{S}_k} \left [ U_\alpha \right  ],
\]
where
\[
U_\alpha  = \sum_{\mathbf{x}\in [d]^k} \ketbra{x_{\alpha^{-1}(1)}, ..., x_{\alpha^{-1}(k)}}{x_1, ..., x_k}.
\]

We can thus write

\begin{align}
    F &= \text{Tr} \bigg ( \mathop{\mathbbm{E}}_{S, T} \bigg [ \left ( \Pi_{d^{|S|}}^k \otimes_S \Pi_{d^{|S^c|}}^k \right ) \left (  \Pi_{d^{|T|}}^k \otimes_T \Pi_{d^{|T^c|}}^k \right ) \bigg ] \bigg ) \\
    & = \mathop{\mathbbm{E}}_{\substack{S, T \subseteq [n] \\ \alpha, \beta, \gamma, \delta \in \mathcal{S}_k}}  \left [  \text{Tr} \left ( \left (  U^S_\alpha \otimes_S U^{S^c}_\beta ~ \right ) \left ( ~  U^T_\gamma \otimes_T  U^{T^c}_\delta  \right ) \right ) \right ]  \\
    & =\mathop{\mathbbm{E}}_{\substack{S, T \\ \alpha, \beta, \gamma, \delta}} \bigg [  \text{Tr} \bigg (  U^{S\cap T} _{\alpha \gamma} \otimes  U^{S\cap T^c} _{\alpha \delta} ~ \otimes ~   U^{S^c \cap T} _{\beta \gamma} ~ \otimes ~  U^{S^c \cap T^c } _{\beta \delta} \bigg ) \bigg ] \\
& =\mathop{\mathbbm{E}}_{\substack{S, T \\ \alpha, \beta, \gamma, \delta}} \bigg [  \text{Tr} \bigg (  U^{S\cap T} _{\alpha \gamma} \bigg ) ~  \text{Tr} \bigg (  U^{S\cap T^c} _{\alpha \delta} \bigg  ) ~ \text{Tr} \bigg (  U^{S^c \cap T} _{\beta \gamma} \bigg )  ~ \text{Tr} \bigg (  U^{S^c \cap T^c   } _{\beta \delta} \bigg )  \bigg ],
\end{align}
where the superscripts denote the systems that the unitaries act on.

We now use the following fact, giving proof in \cref{app:cycles}.
\begin{restatable}{fact}{cycles} \label{fact:cycles}
For some permutation $\alpha \in \mathcal{S}_k$, consider the unitary
\[
U_\alpha  = \sum_{\mathbf{x}\in [d]^k} \ketbra{x_{\alpha^{-1}(1)}, ..., x_{\alpha^{-1}(k)}}{x_1, ..., x_k}.
\]
Let $c(\alpha)$ denote the number of cycles in the permutation $\alpha$ when written in cycle decomposition. Then we have
    \[
    \text{Tr} \left ( U_\alpha \right ) = d^{c(\alpha)}.
    \]
    
\end{restatable}

Hence we can write
\begin{align}
    F & =\mathop{\mathbbm{E}}_{\substack{S, T \\ \alpha, \beta, \gamma, \delta}} \bigg [  \text{Tr} \bigg (  U^{S\cap T} _{\alpha \gamma} \bigg ) ~  \text{Tr} \bigg (  U^{S\cap T^c} _{\alpha \delta} \bigg  ) ~ \text{Tr} \bigg (  U^{S^c \cap T} _{\beta \gamma} \bigg )  ~ \text{Tr} \bigg (  U^{S^c \cap T^c   } _{\beta \delta} \bigg )  \bigg ] \\
  &= \mathop{\mathbbm{E}}_{\substack{S, T \\ \alpha, \beta, \gamma, \delta}} \bigg [  d^{\abs{S \cap T} c(\alpha \gamma) + \abs{S \cap T^c} c(\alpha \delta) + \abs{S^c \cap T} c(\beta \gamma) + \abs{S^c \cap T^c} c(\beta \delta)} \bigg ].
\end{align}

We can simply this expression slightly and eliminate one of the sums over $\mathcal{S}_k$ as follows. First we use the substitutions relabelling $\delta' = \beta \delta$ and $\gamma' = \beta \gamma$:

\begin{align}
    F &= \mathop{\mathbbm{E}}_{\substack{S, T \\ \alpha, \beta, \gamma', \delta'}} \bigg [  d^{\abs{S \cap T} c(\alpha \beta^{-1} \gamma') + \abs{S \cap T^c} c(\alpha \beta^{-1} \delta') + \abs{S^c \cap T} c(\gamma') + \abs{S^c \cap T^c} c(\delta')} \bigg ] .
\end{align}
Next we can define $\alpha' = \alpha \beta^{-1}\delta'$, followed by $\delta '' = \delta^{-1}$ to get
\begin{align}
    F &= \mathop{\mathbbm{E}}_{\substack{S, T \\ \alpha', \beta, \gamma', \delta'}} \bigg [  d^{\abs{S \cap T} c(\alpha \delta'^{-1} \gamma') + \abs{S \cap T^c} c(\alpha') + \abs{S^c \cap T} c( \gamma') + \abs{S^c \cap T^c} c( \delta')} \bigg ] \\
    &= \mathop{\mathbbm{E}}_{\substack{S, T \\ \alpha', \gamma', \delta''}} \bigg [  d^{\abs{S \cap T} c(\alpha \delta'' \gamma') + \abs{S \cap T^c} c(\alpha') + \abs{S^c \cap T} c( \gamma') + \abs{S^c \cap T^c} c( \delta''^{-1})} \bigg ].
\end{align}

Finally, we can use the fact that the cycle number is preserved under inverses, i.e. $c(\delta^{-1}) = c(\delta)$, and perform a global relabelling to arrive at
\begin{align}
    F &= \mathop{\mathbbm{E}}_{\substack{S, T \\ \alpha, \gamma, \delta}} \bigg [  d^{\abs{S \cap T} c(\alpha \delta \gamma) + \abs{S \cap T^c} c(\alpha) + \abs{S^c \cap T} c( \gamma) + \abs{S^c \cap T^c} c( \delta)} \bigg ].
\end{align}

We now separate out the case where $\alpha, \delta, \gamma$ are all the identity permutation $e \in \mathcal{S}_k$, so here for the cycle number is $c(\alpha) =  c(\delta) =  c(\gamma) = k$.

\begin{align}
     F &= \mathop{\mathbbm{E}}_{\substack{S, T \\ \alpha, \gamma, \delta}} \bigg [  d^{\abs{S \cap T} c(\alpha \delta \gamma) + \abs{S \cap T^c} c(\alpha) + \abs{S^c \cap T} c( \gamma) + \abs{S^c \cap T^c} c( \delta)} \bigg ] \\
    &=  \tfrac{1}{(k!)^3}\mathop{\mathbbm{E}}_{\substack{S, T}} \bigg [  d^{k (\abs{S \cap T} + \abs{S \cap T^c} + \abs{S^c \cap T} + \abs{S^c \cap T^c} )} \bigg ] \\ &\quad +  \tfrac{(k!)^3 -1 }{(k!)^3} \mathop{\mathbbm{E}}_{\substack{S, T  \\ \alpha, \delta, \gamma \in \mathcal{S}_k \\ (\alpha, \delta, \gamma) \neq (e,e,e) }} \bigg [  d^{\abs{S \cap T} c(\alpha \delta \gamma ) + \abs{S \cap T^c} c(\alpha) + \abs{S^c \cap T} c(\gamma ) + \abs{S^c \cap T^c} c(\delta)} \bigg ] \\
    &\leq  \frac{d^{nk}}{(k!)^3}  + \mathop{\mathbbm{E}}_{\substack{S, T  \\ \alpha, \delta, \gamma \in \mathcal{S}_k \\ (\alpha, \delta, \gamma) \neq (e,e,e) }} \bigg [  d^{\abs{S \cap T} c(\alpha \delta \gamma ) + \abs{S \cap T^c} c(\alpha) + \abs{S^c \cap T} c(\gamma ) + \abs{S^c \cap T^c} c(\delta)} \bigg ] .
    \end{align}

We now seek an upper bound on the second expression. If $\alpha$, $\delta$, $\gamma$ are all not the identity, then at least two of $c(\alpha \delta \gamma )$, $c(\alpha)$, $c(\gamma )$,  $c(\delta)$, must be at most $k-1$ (it is clearly not possible for only one of them to be $k-1$, and the others $k$). We now find the maximum amongst these ${4 \choose 2}=6$ possibilities.

One can verify that the symmetries of $S \leftrightarrow S^c$ and $T \leftrightarrow T^c$ in the expectation value mean we can consider without loss of generality the following two cases
\begin{align}
    \mathop{\mathbbm{E}}_{S,T} \bigg [  d^{\abs{S \cap T}k + \abs{S \cap T^c}k + \abs{S^c \cap T} (k-1) + \abs{S^c \cap T^c} (k-1)} \bigg ] &= d^{n(k-1)}\mathop{\mathbbm{E}}_{S,T} \bigg [ d^{\abs{S\cap T} + \abs{S\cap T^c}} \bigg ] \\
        \mathop{\mathbbm{E}}_{S,T} \bigg [  d^{\abs{S \cap T} k + \abs{S \cap T^c}(k-1) + \abs{S^c \cap T} (k-1) + \abs{S^c \cap T^c} k} \bigg ] &= d^{n(k-1)}\mathop{\mathbbm{E}}_{S,T} \bigg [ d^{\abs{S\cap T} + \abs{S^c\cap T^c}} \bigg ].\label{eq:second_ST}
\end{align}
The first expression can be evaluated as follows

\begin{align}
     \mathop{\mathbbm{E}}_{\substack{S, T }} \bigg [  d^{ \abs{S \cap T}  + \abs{S \cap T^c}} \bigg ] &=   \mathop{\mathbbm{E}}_{\substack{S}} \bigg [  d^{ \abs{S}}\bigg ]  \\
     &=  \frac{1}{2^n} \sum_{s=0}^n {n \choose s} d^{s} \\
     &= \left (\frac{1+d}{2} \right )^n.
\end{align}

For the second expression, in \cref{eq:second_ST}, we can rewrite the expectation value over $S$ and $T$ by introducing random variables $X_i$, for $i \in [n]$, that equal $1$ if $i \in S$ and $0$ otherwise (each occurring with probability $\frac{1}{2})$. Let $Y_i$ be similarly defined with respect to $T$. We can then write

\begin{align}
     \mathop{\mathbbm{E}}_{S,T} \bigg [ d^{\abs{S\cap T} + \abs{S^c\cap T^c}} \bigg ] &=  \mathop{\mathbbm{E}}_{ \substack{X_1, \dots, X_n \\ Y_1, \dots, Y_n}}  \bigg [ d^{\sum_i X_i Y_i + (1-X_i)(1-Y_i)} \bigg ] \\
     &=  \prod_{i=1}^n \mathop{\mathbbm{E}}_{ X_i, Y_i}  \bigg [ d^{X_i Y_i + (1-X_i)(1-Y_i)} \bigg ] \\
     &= \bigg ( \frac{1+d}{2} \bigg )^n,
\end{align}
where in the second line we used the independence of the random variables $X_i$ and $Y_i$.

Having considered all possibilities, we can conclude that the following inequality holds
\[
\mathop{\mathbbm{E}}_{\substack{S, T  \\ \alpha, \delta, \gamma \in \mathcal{S}_k \\ (\alpha, \delta, \gamma) \neq (e,e,e) }} \bigg [  d^{\abs{S \cap T} c(\alpha \delta \gamma ) + \abs{S \cap T^c} c(\alpha) + \abs{S^c \cap T} c(\gamma ) + \abs{S^c \cap T^c} c(\delta)} \bigg ]  \leq d^{(k-1)n} \bigg ( \frac{1+d}{2} \bigg )^n.
\]

This places the following bound on $F$
\begin{align}
    F &\leq  \frac{d^{nk}}{(k!)^3}  + d^{(k-1)n} \left (\frac{1+d}{2} \right )^n \\
    &= \frac{d^{nk}}{(k!)^3} + d^{nk} \left (\frac{1+d}{2d} \right )^n .
\end{align}

Bringing this all together and plugging in our bound on $F$ into \cref{eq: F expression}, we have
\begin{align}
D^2 &\leq \frac{k!}{4} \bigg ( F \cdot \frac{(k!)^3}{d^{nk}} -  e^{-k^2 / d^{n}}   \bigg  ) \\
& \leq  \frac{k!}{4} \bigg (1 + (k!)^3\left (\frac{1+d}{2d} \right )^n -  e^{-k^2 / d^{n}}   \bigg  ),
\end{align}

as claimed.

\end{proof}

\section{Concluding Remarks}

We have demonstrated that testing bipartite productness requires at least $\Omega(n / \log n)$ copies, which matches the upper bound of $O(n)$ from \cite{harrow2017sequential} up to a logarithmic factor. As testing multipartite productness only requires $O(1/\epsilon^2)$ copies \cite{harrow2013testing, montanaro2013survey}, our result is the first to show that testing bipartite productness is strictly harder. As discussed in the introduction, this also implies that computing the generalised geometric measure of entanglement for multipartite states requires $\Omega(n / \log n)$ copies in general. Another implication is that if one wishes to test the property of some graph state corresponding to a non-connected graph using less than $\Omega(n / \log n)$ copies, one would need a different approach to that of simply testing for bipartite productness.

In would be interesting to see if our bound could be further tightened to $\Omega(n)$ to match the known upper bound more closely, although we believe alternative proof techniques would be needed. One could also study the dependence on $\epsilon$ in more depth -- in our techniques this dependence appears via \cref{lem:schmidtprob} and \cref{fact:schmidt} in combination, however in the limit of large $n$ the relevant term in \cref{lem:existenceimplies} tends to zero for all $\epsilon$. 

Another compelling avenue would be to examine the complementary property of being genuine multipartite entangled, which to the best of our knowledge has not yet been studied. Clearly one cannot directly test for this in the property testing framework, as any bipartite product state can be arbitrarily close to a GME state in trace distance, by slight perturbations of the state. However, one could consider the property of being maximally multipartite entangled according to some measure, such as the generalised geometric measure discussed in this work.

\subsection*{Acknowledgements}

BDMJ is grateful for support from UK EPSRC (EP/SO23607/1). AM acknowledges funding from the European Research Council (ERC) under the European Union’s Horizon 2020 research and innovation programme (grant agreement No. 817581).

\addcontentsline{toc}{section}{References}
\bibliographystyle{alphaurl}
\bibliography{references}

\appendix

\addtocontents{toc}{\protect\setcounter{tocdepth}{1}}

\section{Additional Proofs}
\subsection{Proof of \cref{fact:schmidt}} \label{app:schmidt}

\schmidt*

\begin{proof}~
\begin{enumerate}[(i)]
    \item Write $\ket{\psi} = \sum_i \gamma_i \ket{u_i}\ket{v_i}$ in Schmidt decomposition, with $\ket{u_i}$ and $\ket{v_i}$ respectively orthonormal sets and $\gamma_i$ non-negative and non-increasing with $i$. Denote $\gamma_{\max} \equiv \gamma_0$ as the maximum Schmidt coefficient. Clearly taking $\ket{\alpha} = \ket{u_0}$ and $\ket{\beta} = \ket{v_0}$ shows that
\[
    \gamma_{\max} \leq {\max}_{\alpha, \beta} \abs{\bra{\psi} ~ \ket{\alpha} \ket{\beta}}.
\]
We also have that

\begin{align}
    {\max}_{\alpha, \beta} \abs{\bra{\psi} ~ \ket{\alpha} \ket{\beta}}  &\leq \sum \gamma_i \abs{\braket{\alpha}{u_i} \braket{\beta}{v_i}} \\
    &\leq \gamma_{\max}  \sum \abs{\braket{\alpha}{u_i} \braket{\beta}{v_i}} \\
    &\leq \gamma_{\max}  \sqrt{\sum_i \abs{\braket{\alpha}{u_i}}^2 \sum_j \abs{\braket{\beta}{v_j}}^2} \\
    &\leq \gamma_{\max},
\end{align}
where we used the Triangle and Cauchy-Schwarz inequalities.

\item 
Recall the well known relation between trace distance and fidelity for pure states
\[
    \tfrac{1}{2}\norm{\ketbra{\psi} - \ketbra{\phi}}_1 =  \sqrt{1 - \abs{\braket{\psi}{\phi}}^2}.
\]

    Now let $\ket{\phi} = \ket{\alpha}\ket{\beta}$ be a BP state (written across some bipartition). Then for $\gamma$ the maximum Schmidt coefficient of $\ket{\psi}$ over all bipartitions, we have
    \[
    \tfrac{1}{2}\norm{\ketbra{\psi} - \ketbra{\phi}}_1 =  \sqrt{1 - \abs{\bra{\psi} ~ \ket{\alpha}\ket{\beta}}^2} \geq \sqrt{1 - \gamma^2},
    \]
    where the last inequality follows from part $(i)$.
    \end{enumerate}

\end{proof}

\subsection{Proof of \cref{fact:haar}} \label{app:haar} 

\haarcloseness*

\begin{proof}

We can write 

\[
\rho' = \frac{1}{1-p} \int d \phi ~ \ketbra{\phi}^{\otimes k} \mathbf{1}_S (\phi),
\]

where $\mathbf{1}_S (\phi) = 1$ if $\ket{\phi} \in S$ and $0$ otherwise (the indicator function), and $p$ is a normalisation factor enforcing $\text{Tr}(\rho')=1$. Then we have 

\begin{align}
    \norm{\rho - \rho'}_1 &= \norm{\int d\psi ~ \ketbra{\psi}^{\otimes k}  \bigg (\id  - \frac{1}{1-p} \mathbf{1}_S \bigg )}_1 \\
    &\leq \int d\psi ~ \norm{\ketbra{\psi}^{\otimes k}}_1  \abs{1  - \frac{1}{1-p} \mathbf{1}_S } \\
    &= \int d\psi ~ \abs{1  - \frac{1}{1-p} \mathbf{1}_S } \\
    &= \int_S d\psi ~ \left ( \frac{1}{1-p} - 1  \right ) + \int_{S^c} d\psi \\
    &= (1-p)  \left ( \frac{1}{1-p} - 1  \right ) + p \\
    &= 2 p.
\end{align}
So $D(\rho, \rho') \equiv \frac{1}{2} \norm{\rho - \rho'}_1 = p$, as required.

\end{proof}

\subsection{Proof of \cref{fact:binbounds}} \label{app:binbounds}

\binbounds*

\begin{proof}
    Firstly, we have
    \begin{align}
{a + b -1 \choose b} &= \frac{(a + b - 1)!}{b! (a -1)!} \\
&= \frac{(a + b - 1) ~ \dots ~ (a)}{b!} \\
& \geq  \frac{a^b}{b!}.
\end{align}
For the upper bound, observe that

\begin{align}
{a + b -1 \choose b} &= \frac{(a + b - 1)!}{b! (a -1)!} \\
&= \frac{(a + b - 1) ~ \dots ~ (a)}{b!} \\
&\leq \frac{(a + b - 1)^b}{b!} \\
&\leq \frac{a^b}{b!} (1   + \frac{b}{a})^b \\
&\leq \frac{a^b }{b!} e^{b^2 / a},
\end{align}
where in the last line we used that $1+x \leq e^x$ for all real $x$.
\end{proof}

\subsection{Proof of \cref{fact:cycles}} \label{app:cycles}

\cycles*

\begin{proof}

    The trace of a permutation matrix is the number of fixed points. First consider if there is only one cycle. Then there are $d$ fixed points, occurring exactly when $x_1 = \dots = x_k$. Now suppose $\alpha$ has $c(\alpha)=m$ cycles, and write it in cycle decomposition as
    \[
    \alpha = \alpha_1 \dots \alpha_m.
    \]
    As these cycles act on independent copies of the system, we can write
    \[
    U_\alpha = U_{\alpha_1} \otimes \dots \otimes U_{\alpha_m},
    \]
    from which it follows that
    \begin{align}
        \Tr(U_\alpha) &= \text{Tr}(U_{\alpha_1}) \dots \text{Tr}(U_{\alpha_m}) \\
        &= d \times \dots \times d \\
        &= d^m.
    \end{align}

\end{proof}

\addtocontents{toc}{\protect\setcounter{tocdepth}{2}}

\end{document}